\documentclass{sig-alternate}
\usepackage{graphicx}
\usepackage{url}
\usepackage{amsmath}

\begin{document}
\conferenceinfo{}{}
\title{Wild Card Queries for Searching Resources on the Web}

\numberofauthors{2}
\author{
\alignauthor Davood Rafiei\\ 
  \affaddr{Computing Science Department}\\
  \affaddr{University of Alberta}\\
  \email{drafiei@cs.ualberta.ca}
\alignauthor Haobin Li\\
  \affaddr{Computing Science Department}\\
  \affaddr{University of Alberta}\\
  \email{haobin@cs.ualberta.ca}
}

\maketitle

\newtheorem{definition}{Definition}
\newtheorem{theorem}{Theorem}
\newtheorem{corollary}{Corollary}
\newtheorem{lemma}{Lemma}

\newcommand{\bs}{\bigskip}
\newcommand{\ms}{\medskip}
\newcommand{\s}{\smallskip}
\newcommand{\n}{\noindent}

\begin{abstract} 

We propose a domain-independent framework for searching and retrieving 
facts and relationships within natural language text sources.
In this framework, an extraction task over a text collection
is expressed as a query that  combines text fragments with wild cards, 
and the query result is a set of facts in the form of unary, 
binary and general $n$-ary tuples.
A significance of our querying mechanism is that, despite being both
simple and declarative, it can be applied to a wide range of extraction tasks.
A problem in querying natural language text though is that a user-specified 
query may not retrieve enough exact matches. Unlike term queries which 
can be relaxed by removing some of the terms (as is done in search engines), 
removing terms from a wild card query without ruining its meaning is 
more challenging.
Also, any query expansion has the potential to introduce false positives. 
In this paper, we address the problem of query expansion, and also analyze 
a few ranking alternatives to score 
the results and to remove false positives.
We conduct
experiments and report an evaluation of the effectiveness of our 
querying and scoring functions. 

\end{abstract}





\section{Introduction}
\label{sec:intro}
The World Wide Web contains a vast amount of information and is a rich source 
for data extraction, but manually extracting data from the Web is a tedious 
and time consuming process, especially when a large amount of data matches the 
extraction criteria. Example extraction tasks include compiling a list of 
Canadian writers, finding a list of medications for a disease,
etc. Unless such lists have already been compiled and made available on the 
Web, one has to query a search engine, examine the pages returned, and extract 
a handful of instances from each page (if there is any at all). 
The problem is further complicated by the flexibility of natural languages.
Consider the example of extracting \emph{Canadian writers}; many bona fide writers are not referred to as writers. Instead, they are often coined as \emph{authors}, \emph{novelists}, \emph{journalists}, etc. If only the phrase ``Canadian writers'' is used in the query, many qualified instances will not be extracted, thus the extraction quality is compromised.
Many previous data extraction systems either focus on a more specific task 
by imposing tight restrictions on the type of data that can be extracted
(e.g. finding course offerings and job postings) or are only applicable to
documents that follow a specific formatting (e.g. wrappers). 
For example, the KnowItAll~\cite{knowitallshort} system can extract hyponyms of a user-specified class. The online prototype of the system
is further extended to support a few more specific binary relations (e.g.
{\it X ``ceo of'' Y}). 
A challenge facing many data extraction systems in general is that
the extraction task often is not easy to define or a definition may not be 
accurate, leading to low precision and recall.
Limiting the extraction task to a few predefined classes is one
way to reduce the complexity of the problem.

In this paper, we address the problem by introducing
a framework that allows an extraction task to be encoded as a simple query.
A query in this framework is a natural language sentence or 
phrase\footnote{We assume query phrases are
in English, but our framework should be applicable to other languages as well.} 
with some wild cards, and the result
of a query is a ranked list of matching tuples. 
For instance, given 
the query ``\% is a car manufacturer'', the output is expected to
be a ranked list of car manufacturers, preferably the real car manufacturers 
ranked the highest. This query only uses one wild card, here denoted with
{\it \%}. In general, a query can use more than one \% wild card, and 
the result of the query in this case is a table with one column for each 
occurrence of the wild card. 

Our first contribution is a declarative querying framework
for data extraction. Integration of wild cards in our queries
can generally reduce the number of queries that must
be issued, hence simplifying the extraction task.
In our earlier example about 
\emph{Canadian writers}, for instance, a user can use one type of wild 
card to indicate that terms similar to \emph{writers} should also 
be considered. Another type of wild card may be used to indicate 
a probable position of the desired data, from which values 
can be extracted. Combining such wild cards
with natural language phrases can provide a simple but powerful interface,
which can handle much more extraction tasks than previous systems.
There is a close correspondence between our queries and star-free regular 
expressions; our queries make use of certain abstractions geared 
toward natural languages which make it simpler to write queries.
Since the result of a query is a relation, our queries with some 
syntactic sugar can be integrated in the {\it from} clause of SQL queries.
 
Our second contribution is the idea of query expansion through a set of
declarative rewriting rules between paraphrases. 
Since a given user query may not retrieve an adequate number of facts,
rewriting rules are generally expected to improve
the coverage of the queries and the quality of the results.
Our experiments, as reported in Section~\ref{sec:exp}, show that increasing the
number of rewritings can improve both recall and precision.

As our third contribution, we address the problem of ranking 
in the context of rewriting rules. We propose a few algorithms for
ranking the extracted results where each algorithm exploits some of the
relationships that exist between the set of matching tuples and
also between the set of query patterns.
We use the general term {\it pattern} to refer to both query and query rewriting.
When the results are ranked, it is possible to set a cutoff threshold to filter 
invalid rows from the result, making the final results more accessible to the user.

Finally, as our last contribution, we theoretically analyze our
algorithms and experimentally evaluate their behaviour in the setting
of the Web. Our experimental results include comparisons 
with some alternative approaches to data extraction. 

The rest of the paper is organized as follows. 
Section~\ref{sec:wildcard} describes both the syntax and the semantics of 
wild cards, as well as the queries in our framework.
An overview of our query evaluation in the context of the Web is given in 
Section~\ref{sec:overview}. Section~\ref{sec:rewriting}
discusses the details of our rewriting rules and patterns.
Our ranking algorithms are discussed in Section~\ref{sec:ranking}.
Experimental results are presented in Section~\ref{sec:exp}
and the related work is reviewed in Section~\ref{sec:related}.
We end the paper with conclusions and future work in Section~\ref{sec:con}.

\section{Wild Card Queries}
\label{sec:wildcard}

The use of wild cards is prevalent in many areas of computer science, with
examples found in SQL, operating system shells and scripting languages such 
as Perl, Awk and Python. Unlike many of these systems, our introduced wild 
cards iterate over the domains of parts of speech or other meaningful 
natural language word groupings. In particular,
we introduce two types of wild cards, namely * and \%. 

\medskip
\noindent 
{\bf \% wild card}:  The \% wild card represents one or more noun phrases. 
A noun phrase may consist of one or multiple words; for instance, ``movie'' 
and ``action movie'' are both noun phrases. This wild card, when used
in a query, indicates the location of a noun phrase or noun phrases that 
should be extracted.
For example, the query ``summer movies such as \%'' will extract noun 
phrases \emph{Harry Potter}, \emph{Shrek}, and \emph{Spiderman} from the 
following sentence: ``Popular summer movies such as Harry Potter, Shrek and 
Spiderman appeal to audience of all ages.'' 

\medskip
\noindent
{\bf * wild card}: The * wild card represents a set of phrases 
with the same or similar meanings to a given phrase. 
Consider the task of finding a listing of summer movies; we may type
the query ``\% is a summer movie''. However, some bona fide movies are 
often referred to as ``films'', ``blockbusters'', and so on. 
In a naive approach, one may have to try other terms similar 
to ``movie'' manually, save the results each 
time, and put the results together at the end. The naive method is tedious 
and inefficient. In our queries, a term may be enclosed within a pair of * 
to instruct that the search should be extended to include terms and 
phrases similar to the given one.
For example, the user can re-formulate the query as ``\% is a summer *movie*'', 
and the query will be automatically expanded to include related 
queries such as ``\% is a summer film'', ``\% is a summer blockbuster'', etc. 

It is feasible to consider other wild cards. For instance, we could have 
wild cards that only match verbs, adjectives, or a union of nouns, verbs 
and adjectives.
It is also possible to have a wild card that matches a prefix or suffix of
a term or a fixed-length sequence of terms.
In an attempt to keep the syntax of our queries simple, 
this paper only considers the two wild cards
\% and *, as discussed above. The following is a list of example queries:
 
\begin{itemize}
\item \% is a *country*
\item \% is a summer *blockbuster*
\item \% invented the light bulb
\item Google *acquired* \%
\end{itemize}

A query may use any number of wild cards. Given a query with {\it k} 
\% wild cards, the result of the query is a table
with {\it k} columns, one for each \% wild card. 
It is natural to assume that the result of a query is ranked and each row is 
assigned a score as an indication of the level of support the row receives. 
This score may depend on the size and the coverage of the
text collection that is being queried and the set of query rewritings 
that are being used.
Section~\ref{sec:ranking} discusses a few measures to rank the matching tuples 
of a query.

A query can have any number of * wild cards. Given a query $q$ with some * wild
cards, let $q_1$, $\ldots$, $q_s$ be the set of queries that are obtained by 
replacing each * wild card with {\it similar terms}. A row matches $q$ if it matches 
at least one of $q_1$, $\ldots$, $q_s$; the score of the row for $q$ may be
an aggregation 
of the scores of the row for $q_1$, $\ldots$, $q_s$.
For our purpose, two terms are considered {\it similar} if they have the 
same meanings (e.g. synonyms), one is a generalization of the other,
or the two terms can be used  interchangeably in the same context.
The similar terms can be often obtained 
from dictionaries, thesaurus, online corpus \cite{lin97using}, etc. 

\section{Evaluating Wild Card Queries -- An Overview}
\label{sec:overview}

This section provides an overview of evaluating wild card queries over
a text corpus. Without loss of generality, our discussion in this section
is centered 
around the Web and uses the query ``\% is a *blockbuster*'' as  an example.
In particular, we consider the scenario where the extraction
engine is built on top of a search engine.
Naturally the same steps can be taken when the source data 
is stored locally, with a difference that a local collection may be
better indexed and the queries can be 
better optimized. To limit the scope of the paper, we do not 
address the issues related to indexing and query optimization.

\s
\n {\bf Step 1 - query flattening}
As the first step, the query is analyzed and expanded by replacing
the words enclosed by pairs of * wild cards (if any) with their similar terms. 
In the given query, the word ``blockbuster'' is enclosed by *'s; 
similar words to blockbuster, based on an online system \cite{lin97using},
are ``movie'' and ``film''. Two new queries, ``\% is a movie'' and
``\% is a film'' are formed and added to the expanded query set. 
In general, more queries may be added if multiple synonyms are found.

\s
\n {\bf Step 2 - query rewriting}
In the next step, each query in the query set is passed to a Part-Of-Speech (POS) tagger.
Each tagged query is compared with a set of precompiled patterns for possible 
rewritings. The result of query tagging is not always reliable, in particular 
for short queries. To account for those cases, queries are also rewritten
using rules that do not require tagging.
Let's consider the query
``\% is a movie'' first; after tagging, the query conforms to the pattern 
``NP1 is a(n) NP2'' where NP stands for a noun phrase; note that
the wild card \% matches a noun phrase, as we defined earlier.
A pattern may be found relevant to a pre-determined class, based on a
specific relationship it describes, and may be rewritten by other
patterns in the same class.
The pattern ``NP1 is a NP2'', in particular, 
belongs to the \emph{hyponym} class, since the template indicates that NP1 
is a (hyponym of) NP2. Other patterns in the hyponym class include 
``NP2 such as NP1List'', ``NP2 including NP1List'', etc. All patterns 
in the matching class (i.e. the hyponym class) are instantiated according 
to the matched query. Thus, the query set is expanded with extra queries like ``movies such as \%'' and ``movies including \%''. 
Section~\ref{sec:rewriting} discusses our query rewritings in more detail.
Similarly, the query ``\% is a blockbuster'' also matches a pattern in the 
hyponym class, and the query set is further expanded. If the query cannot 
match any pattern, no query expansion will occur at this step. 

\s
\n {\bf Step 3 - information retrieval engine}
As the third step, all queries in the query set are sent to a search engine. For each 
query, the matching snippets are downloaded for 
further processing. When there is a large number of matches, 
only a fixed number of them are selected. 
HTML tags are stripped from downloaded snippets for each query and the remaining text
is analyzed to identify the pieces that match the query. 
Noun phrases that appear 
in the positions of \% wild cards of a query are extracted from the text 
and are saved in the result set. Words other than noun phrases should not 
be extracted even if they appear in target locations. Suppose the query 
``\% invented the light bulb'' is sent to a search engine and the following 
two snippets are among those returned.
\begin{itemize}
\item Thomas Edison is often said to have \emph{invented the light bulb}.
\item We all learned in our history classes that Thomas Edison \emph{invented the light bulb} in 1879.
\end{itemize}
The POS tagger identifies that the word ``have'' in the first snippet is not 
a noun phrase, while ``Thomas Edison'' from the second snippet is. 
Therefore, the phrase ``Thomas Edison'' is extracted but ``have'' is not.

\s
\n {\bf Step 4 - relevance ranking}
The result of extraction in the previous step is a set of rows; for the given
example, each row is a noun phrase. 
A ranking algorithm is applied to the extracted set.
Section~\ref{sec:ranking} gives the details of our ranking algorithm.
Finally, a sorted list of rows is returned.
The rest of this paper will focus on {\it query expansion} and 
{\it relevance ranking}.

\section{Rewriting Queries}
\label{sec:rewriting}

A challenge in querying 
and information retrieval from 
natural language 
text is the possibility of a mismatch between the expressions of queries 
and texts that have the relevant information. 
A fact can potentially be expressed in many different contexts and 
a query that gives one context can miss many qualified candidates.
We propose rewriting rules to express the set of transformations that 
can be applied to a wild card query leading to alternative query expressions
that are expected to return the same or semantically related results but are
syntactically different.
Query rewriting is expected to increase recall,
as it can be seen from our example in Section~\ref{sec:overview}.
Our experiments in Section~\ref{sec:exp}
confirm that there is also a correlation between the number of rewritings and
the precision of the retrieved results. There are more evidence 
in favor of rewriting natural language questions.
For example, at TREC-10 QA evaluation, the winning system used an 
extensive set of rewritings as its only resource \cite{voorhees01overview}.

\subsection{Rewriting Rule Language}

The rewriting rule language lists different ways of rewriting a query.
Each rule here is of the form \emph{rule-head $\to$ rule-body}. 
A rule head consists of one or more regular expressions, and a rule body 
consists of one or more rewritings with place holders.
Multiple regular expressions in the head or
rewritings in the body can be listed each in a new line.
A rule matches a query if any one of the regular expressions in the head 
matches the query.
When a rule matches a query, the query is expanded with all 
rewritings in the rule body. 
For each match, keywords from the query 
may be remembered using {\it capturing groups} (e.g. parentheses)
and the remembered values may be recalled using back references in the 
rule body.
The remembered values can be transformed (e.g. from a singular noun to a 
plural noun) before being used in the rewritings.
Each transformation is usually done by looking up a table, from a set of
tables compiled in advance from a dictionary, thesaurus, and other 
offline sources.
This allows us to write generic rewritings that can match a large number
of queries. 
Our rewriting rule language is extendable and one can add more rules 
as they become available.
Here is an example of a rule. Given the query ``movies such as \%'',
the following rule generates ``\%, and other movies'' and ``\% is a movie'' as
possible rewritings.
{\tt
\begin{tabbing}
(.+),? such as (.+)\\
(.+),? including (.+)\\
$\rightarrow$\\
\$2, and other \$1	\= \&\& plural(\$1)\\
\$2 is a \$1		\> \&\& singular(\$1)\\
\end{tabbing}
}
Let $n$ denote the average number of rewritings produced for each
query. If a query uses $k$ star wild cards and each of these wild cards
is replaced with $m>0$ similar terms on average, the query expansion would 
produce $m^k(n+1)$ queries. 
To appreciate the power of the rewriting rule language, consider an alternative 
scenario where one has to enumerate and try many of those queries manually.

\subsection{Compiling Rewriting Rules}

Gathering a comprehensive set of rewritings for a large set of 
possible queries in advance can be challenging, especially if not much 
is known about the queries in advance. 
However, many short queries fall into some sort of common templates for which 
generic paraphrasing rules usually apply. In English, two classes of generic 
rules that we can easily
identify are {\it hyponyms} and {\it morphological variants}.
A hyponym pattern describes lexico-syntactic relations that can
be used to infer one element is a hyponym of another within a sentence.
Hearst gives a list of hyponym patterns~\cite{Hearst92}.
A sample of hyponym patterns (either from Hearst's or hand-crafted by us)
can be found in {Table}~\ref{tab:hypo}.

\begin{table}[htb]
\centering
\begin{tabular}{c} \hline
NP1 \{,\} ``such as'' NP2List \\
``such'' NP1 ``as'' NP2List \\
NP1 \{,\} ``especially'' NP2List \\
NP1 \{,\} ``including'' NP2List \\
NP2List ``and other'' NP1 \\
NP2List ``or other'' NP1 \\
NP2, ``a(n)'' NP1 \\
NP2 ``is a(n)'' NP1 \\
NP1 NP2 \\
\hline\end{tabular}
\caption{Hyponym patterns}
\label{tab:hypo}
\end{table}

The morphological variants of verbs are useful for rewriting many queries 
that contain verbs. A given query may be rewritten by simply changing its verb 
tense and without much affecting its meaning.
Many extraction tasks are expressed in the form of 
``Subject transitive-Verb Object'' which can be rewritten in a passive form 
and vice versa. For example, if a user wants to find out who 
invented the light bulb, she can express the extraction as 
``\% invented the light bulb''; a rewriting of the query is 
``the light bulb was invented by \%''.
Our morphological patterns enumerate different 
verb tenses (e.g. present tense, past tense, \ldots) and use both active and 
passive forms. 
We were able to express
all the relationships described by patterns in the hyponym and morphological
classes as rules in our rule language.

Although generic patterns and rewritings can be applied to a wide range of
queries, it is not hard to find queries where no generic pattern is applicable 
or sufficient. With a rule language, many of these rules can be expressed and
regularly revised as more queries are seen.
There is also an active research on automatically generating paraphrases
in more specific settings~\cite{lin01dirt,RavichandranH02}, which may also 
be encoded as rules in our framework. 

\section{Bringing Order to Results}
\label{sec:ranking}
The data extraction process discussed in this paper 
can accumulate a large set of candidates. Not all candidates 
are correct, meaning that a user would not consider them as good matches.

\subsection{Sources of Errors}
A query typed by user can match many non-relevant rows, in the sense that they
may not be anticipated by the user. For instance, the query
``\% is a country'' matches the statement ``Joe is a country singer,''
thus Joe will be added to the list of country names.
Through a heavy natural language processing (NLP), one
may reduce the number of those false positives;
however, these techniques are generally computationally intensive and 
may not scale up well to large volumes of data and queries (e.g. on the Web).
In general, broad queries are likely to match more false positives.
Rewriting queries can also broaden the queries and introduce additional
false positives.
Also false positives galore on the Web;
since the published content may not be verified for correctness,
statements, such as ``New York is the capital of the 
United States'' are not rare. 

One more source of error is due to using a POS tagger.
Although POS taggers produce good results most of the time, the accuracy
usually depends on factors such as the number of words in the lexicon, etc.
Sometimes verbs are mis-classified as noun phrases, or vice versa.
Since correct extractions are inter-mixed with errors, it is important to rank all 
candidates in terms of their relevance to the user query. 

\subsection{Ranking Heuristics}
\label{subsec:heuristics}
Given a query and a set of rewritings, we want to find out 
meaningful ways of ordering the results. 
In this section, we present a few heuristics before we analyze them within 
a more general ranking scheme in the next section.

\medskip
\noindent {\bf Number of Matched Pages or Documents (NPages)}:
The relevant matches of a query are likely to appear frequently within the 
query pattern \footnote{The query would not have been issued in the first
place if it is assumed otherwise.} 
or one of its rewritings. 
One heuristic is to rank a tuple based on the number of pages or documents
in which it matches the query or one of its rewritings. 

\medskip
\noindent {\bf Mutual Information (MI)}
Another ranking scheme which has been used previously to quantify the
relationship between two random variables is the {\it Mutual Information} (MI).
If we denote the probability that a document
contains the text of query {\it q} (ignoring wild cards) with $P(q)$, the probability that
a document contains a candidate {\it r} with $P(r)$, and the probability that
a document contains a proper encoding of {\it r} in {\it q}
with $P(q,r)$, then the mutual information between {\it q} and {\it r}
is defined as 
\[
MI(q,r) = log\frac{P(q,r)}{P(q)P(r)}.
\] 
In some formulations of the mutual information,
the above formula is multiplied by $P(q,r)$~\cite{Chakrabarti2002mtw}.
This measure is used in the past, for instance to evaluate the association between 
words~\cite{Church-Hanks-1989}, and also between the instances of a class 
and a discriminative phrase ~\cite{knowitallshort}. 
In our case, since $P(q)$ is fixed for a given query, the score of a
candidate can be estimated as the ratio $P(q,r)/P(r)$.
For a given query and candidate,
the MI measures the conditional probability that the candidate appears 
within the query template given that the candidate appears in a document.

\medskip
\noindent {\bf Number of Matched Patterns (NPatterns)}
Another simple ranking is to count for each candidate, the number of 
different patterns (including the query and its rewritings) that would 
extract the candidate.
Because of the semantic relationship between a query and its rewritings,
if a candidate is retrieved by multiple patterns, then there is probably
a good indication that it is indeed a good match.

\medskip
\noindent {\bf Discussion of Ranking Heuristics}
Our experiments comparing these heuristics (as reported in Sec.~\ref{sec:exp}) 
show mixed results for different queries.
A general drawback for NPages and NPatterns is that all query patterns are 
treated equally important.
With NPages, some correct instances that are not frequent cannot 
be well-separated from false positives.
Also under NPages, the scores would be inflated when there are
duplicates (such as those on the Web).
A drawback for MI is that a candidate may not appear with the query but 
it may appear with one of its rewritings.
Selecting a single pattern is not guaranteed to achieve a high recall.
Also it is not clear how MI can be extended to account for multiple rewritings  
of a query and also queries with multiple extractors. 

\subsection{Relationship Graph between Patterns and Tuples}
\label{sec:mutual-ranking}

Let $S_P$ denote the set of query patterns (i.e. the user-specified query
and all of its rewritings) and $S_T$ denote the set of matching candidate tuples.
Consider the bipartite graph $G$ formed between $S_P$ and $S_T$ with
an edge from $p \in S_P$ to $t \in S_T$ if $t$ matches $p$ in some text. 
In some settings, the edges of the graph may be assigned weights to 
indicate the degree of a match.
We define a ranking $F$ as a function that maps $S_T$ to an 
$n$-dimensional vector where $n$ is the size of $S_T$.
Some of our earlier heuristics can be seen as special cases of this ranking function.
In particular, $NPatterns(t) = indegree(t)$ and 
\[
NPages(t) = \sum_{p \in S_P} w(p \rightarrow t)
\]
where $w(p \rightarrow t)$ is the number of pages that give rise to a match
between $t$ and $p$.

A limitation of NPatterns and NPages is that all patterns have the same
influence on the scores. Our observations, however, indicate that often there 
are a few patterns that retrieve many good quality rows while the rest 
retrieve many false positives.
To remedy the problem, 
we propose weighting the patterns and propagating the weights 
to the matching candidates. The weights may be assigned at the same time 
rewritings are compiled if they are available and are not expected to change 
much. 
For generic rewritings, however, the weights can change from one query 
to next and a more dynamic weighting scheme may be preferred. 

A hypothesis is that {\it good tuples} and {\it good patterns} 
exhibit a mutually reinforcing relationship: a good tuple is extracted by 
many good patterns; a good pattern extracts many good tuples
\footnote{The same hypothesis is made in a hyperlinked environment where several
ranking algorithms are proposed for finding authoritative Web pages.}.
For example, 
if {\it Canada} is indeed a good match for the query ``\% is a country'', 
it should be extracted by many good related patterns, such as 
``countries including \%'', 
``such countries as \%'', and so on. Similarly, if ``countries such as \%'' 
is indeed a good pattern for extracting country names, it should extract 
many good instances, like {\it the U.S.}, {\it Canada}, {\it China}, etc. 
Our experiments use an adaptation of 
Kleineberg's HITS algorithm \cite{kleinberg99authoritative},
but there are indications that other two-level influence propagation
algorithms can equally  be used.

\medskip
\noindent
{\bf Algorithm PT-hits} 
Let's associate weight $w_T(t)$ to each candidate tuple $t$, and weight $w_P(p)$ 
to each pattern $p$. In an iterative and alternating fashion, the weights
can be updated as follows:
\begin{equation}
w_T(t) = \sum_{\{p|p\,\,{\rm extracts}\,t\}}w_P(p)
\label{eq:T}
\end{equation}
\begin{equation}
w_P(p) = \sum_{\{t|t\,{\rm is\,extracted\,by}\,p\}}w_T(t)
\label{eq:P}
\end{equation}
The weights can initially be all set to 1, then propagated from patterns to
tuples and vice versa. After each iteration, the weights are normalized
to keep them within a bound and also to check the convergence.
It is easy to show that this iteration converges. 

\subsection{Analyses of the Rankings}
\label{sec:ranking-analysis}

Some of the scoring functions presented in Section~\ref{sec:mutual-ranking} 
have desirable properties that are useful in the setting of the Web.
One such property is {\it monotonicity}, i.e. for every pair of 
tuples $t_1$ and $t_2$, if every pattern $p$ that extracts
$t_1$ also extracts $t_2$ and $w(p \rightarrow t_1) \geq w(p \rightarrow t_2)$,
then $w_T(t_1) \geq w_T(t2)$. It is easy to show that the functions 
NPages, NPatterns and PT-hits are all monotone.

Two other interesting properties are {\it stability} and {\it locality}.
A scoring function is stable if small changes in the graph structure between 
patterns and tuples have no effect on the scores of the tuples, and a scoring function
is local if removing an edge has no effect on the weights of non-adjacent tuples.
With data on the Web regularly changing, it is desirable to know if the scores would
change with small changes on the Web.
More formally, let ${\cal G}_{m,n}$ be a collection of bipartite graphs over a set 
of $m$ patterns and $n$ tuples, $F(G)$ denotes the weight vector of
tuples after a scoring function $F$ is applied to a bipartite graph $G$,
and $d()$ be a function measuring the distance between two weight vectors.

\begin{definition}
{\rm A scoring function {\it F} is {\it stable} on ${\cal G}_{m,n}$ under
distance function $d$ if
for every fixed $k$, we have}
\[
\lim_{m\ or\ n \rightarrow \infty} \max_{G \in {\cal G}_{m,n}, 
e_1,\ldots e_k \in E(G)} d(F(G), F(G\backslash \{e_1,\ldots, e_k\})) = 0
\]
{\rm
where $E(G)$ denotes the edges of $G$ and $G\backslash \{e_1,\ldots, e_k\}$ is the
graph obtained after removing edges ${e_1,\ldots, e_k}$ of $G$.
}
\end{definition}

\begin{definition}
{\rm A scoring function {\it F} is {\it local} if for every graph
$G \in {\cal G}_{m,n}$ and every edge $e \in E(G)$, if
$w_1 = F(G)$ and $w_2 = F(G\backslash {e})$, then for all tuples $i$ and $j$ 
which are not adjacent to $e$, the ordering between $i$ and $j$ under $w_1$ and
$w_2$ remains unchanged.
}
\end{definition}

\begin{theorem}
\label{thm:npatt-local}
The scoring function NPatterns is stable under both Kendall tau distance and Manhattan distance
between the weight vectors.
\end{theorem}

\begin{theorem}
\label{thm:npages-local}
NPages is stable under the Kendall tau distance. NPages is also stable under the Manhattan 
distance if $w(p \rightarrow t)$ bound to a fixed constant.
\end{theorem}

\noindent The proofs for Theorems~\ref{thm:npatt-local} and \ref{thm:npages-local} 
can be found in the appendix.

\begin{theorem}
The scoring functions NPatterns and NPages are both local.
\end{theorem}
\begin{proof}
The locality follows from the fact that removing an edge
$e$ only affects the score of a tuple that is linked by $e$.
\end{proof}
Based on the results reported for HITS \cite{BRRT05},
PT-hits is neither stable nor local.

\section{Experiments}
\label{sec:exp}

To evaluate our querying framework and ranking algorithms, we 
built a search tool, called {\it DeWild}, which relies on the Web as its 
source of data
\footnote{{\it DeWild} stands for Data Extraction using Wild cards. The system is available online at
\url{dewild.cs.ualberta.ca}.}.
Using the Web compared to a closed collection has both benefits and drawbacks.
A benefit is that
its information redundancy can compensate for the relatively 
small size and coverage of our rewriting rule set and the lightweight 
NLP techniques used. 
A drawback is that the text collection often is not clean and there are many 
bogus tuples that need to be filtered.

DeWild takes advantage of existing commercial search engines
and currently queries Google 
via its search API.
In our experiments, 200 snippets are downloaded for each extraction pattern
(our online system only downloads 30 snippets per pattern and also lists the set of 
extraction patterns that are tried).  If there 
are fewer than 200 snippets found for a pattern, then all available 
snippets are downloaded\footnote{Our online demo downloads at most 30
snippets for each query and each rewriting to keep the response time short.}.
The snippets returned by 
a search engine typically consists of the search query and its surrounding text.
Since the target data appears immediately before or after the user 
query, they can be often extracted using the snippets only (without 
downloading the actual pages), hence network and processing costs 
are significantly reduced. Though it should be noted that a snippet may lack 
sentence boundaries, and this can reduce POS tagging accuracy.
For our experimental comparisons, 
we implement all heuristics discussed in Section~\ref{sec:ranking}.
However, if it is not explicitly stated otherwise,
DeWild uses PT-hits as its native ranking.

A publicly available POS tagger called 
NLProcessor\footnote{\url{www.infogistics.com/textanalysis.html}} is used 
to identify the part of speech from retrieved text, so that only noun 
phrases are 
extracted for \% wild cards. For * wilds cards, our system uses a collection of related words 
automatically compiled \cite{lin97using} from Wall Street Journal corpus, 
but it can equally use other collections as well.
Next we report our experiments with DeWild. 

\subsection{Effectiveness of the Queries}
To measure the effectiveness of the proposed querying framework in expressing different 
extraction tasks, we used the tool to express and answer natural language questions seeking 
short answers. We randomly selected 100 natural language questions from the AOL query 
log~\cite{aol06}, with 
the selection criteria that the queries were all started with one of the wh-words
{\it what, which, who} and {\it where} and were followed by {\it is, was, was} or {\it were}.
The break-down of the queries were as follows: what-84, who-6, where-7 and which-3.
The {\it why} queries were excluded since these queries often seek
long answers and cannot be expressed in our framework.
Even among the selected queries, not all queries were expected to have short noun phrases 
as answers (e.g. What is a hedge fund? Or, what are the causes of poverty?);
in fact, 58\% of the queries were in this category, and they could not be properly expressed 
as queries. Another 7\% of the queries were ill-formed search queries in the form of
questions which we didn't expect them to have answers (e.g. what is my account balance?).
We removed both classes of queries before further considerations. We further removed another 2\%
of the queries where we couldn't find an answer after an extensive search on the Web, and we
were not sure if the questions had an answer at all. 
The remaining 33\% of the queries were all expressed as queries in our framwork and were 
evaluated using our online tool. 

As a baseline for comparison, we also passed the queries as they were expressed to OpenEphyra~\footnote{http://www.ephyra.info},
a full-fledged open-source question answering system. Similar to ours, OpenEphyra uses
the Web, and in particular the Google search engine, to evaluate the questions, hence both
DeWild and OpenEphyra had access to the same data collection. However, OpenEphyra does much
more extensive parsing of the questions and the answers, and has components for detecting
answer types (of questions) and named entity classes in text.

Top 5 results of each query, as returned by each comparison system, were given 
to two annotators who were asked to check of the set included a correct answer to 
the original question. Table~\ref{table:dewild-openephyra}
shows the precision, averaged over two annotations per question, for both systems. 
DeWild achieves a much better precision, with a good annotator agreement (kappa) of 0.66,
for both short and a bit longer questions, 
retrieving a correct answer for twice as many queries as OpenEphyra. 
\begin{table}
\begin{tabular}{|l|c|c|c|}\hline
           & \multicolumn{3}{|c|}{Question length}\\
           & $< 6$ terms & $\geq 6$ terms & all\\ \hline
DeWild        & 0.88        & 0.41           & 0.68\\
OpenEphyra & 0.41        & 0.22           & 0.32 \\ \hline
\end{tabular}
\caption{Precision of the results for natural language questions}
\label{table:dewild-openephyra}
\end{table}

\subsection{Recall and Precision}

In general, it is difficult to measure recall on the Web
since we often do not know the full answer set. The answer set may not be
all on the Web, or it can be scattered in many pages of which some
may not be crawled or indexed by a search engine.
To measure both recall and precision under these constraints, 
we decided to extract instances of some known classes. 

\begin{table}[htb]
\centering 
\begin{tabular}{|l|l|} \hline
Pattern & Weight \\
\hline \hline
US states, including \%      &  0.2514 \\
US states such as \%         &  0.1826 \\
\% and other US states       &  0.1766 \\
\% is a US state             &  0.0992 \\ 
such US states as \%         &  0.0962 \\
US states, especially \%     &  0.0789 \\
\% or other US states        &  0.0658 \\
\% is the US state 	     &  0.0440 \\
US states \%                 &  0.0052 \\
\%, the US state             &  0 \\ 
US state \%                  &  0 \\ 
\%, a US state               &  0 \\ \hline
\end{tabular} 
\caption{Matching patterns for the target US state names,
with the weights computed by PT-hits}
\label{tab:countries-states}
\end{table} 

In one experiment, we used the names of 50 US states as the
ground truth and tried to retrieve and rank the same data using our
heuristics. 
The query was formulated as ``US states such as \%'';
Table~\ref{tab:countries-states} shows the extraction patterns which matched
our initial query, after instantiating ``US states'' in our generic patterns,
as well as their weights computed by PT-hits.
Clearly, any of the patterns
in the table could have been used as a query and the result returned by DeWild
would have been the same. 
In our evaluation, a retrieved candidate was treated
``correct'' if it was either a full state name or an abbreviation.
Figures~\ref{fig:precision}(a) and (b) show precision and recall 
for our heuristic (PT-hits, NPages, NPatterns) as well as MI and Google 
for comparison. For Google, we tried the queries both with star in
place of the wild card and without, since Google treats star as a special 
character and highlights the matching terms. 
To do a ranking using mutual information (MI), we could use either
the query or one of its rewritings. Since it is not clear which rewriting 
performs the best, we ran the algorithm three times with the discriminative
phrases ``US state of X'', ``US states such as X'', and ``X is a US state''.
For clarity of the figure, we only plotted the one performing the best for 
all recall values. KnowItAll~\cite{knowitallshort} was initially using MI for 
its ranking, but it's configuration was changed over time, hence we ran
our search on the online system
\footnote{\url{www.cs.washington.edu/research/knowitall}} as well. 
The results for KnowItAll must be treated with a grain of salt though, 
since it was the only system using its own local collection.

Many of our comparison systems show a precision of 1 (or very close to 1) 
for small recalls, meaning that they all retrieve a couple of correct state 
names on top. The precision generally drops as more state names are discovered.
PT-hits performs the best for this particular query, retrieving more than 80\%
of the state names without a single error.

\begin{figure*}[t]
\centering
$\begin{array}{cc}
\includegraphics[width=0.47\textwidth]{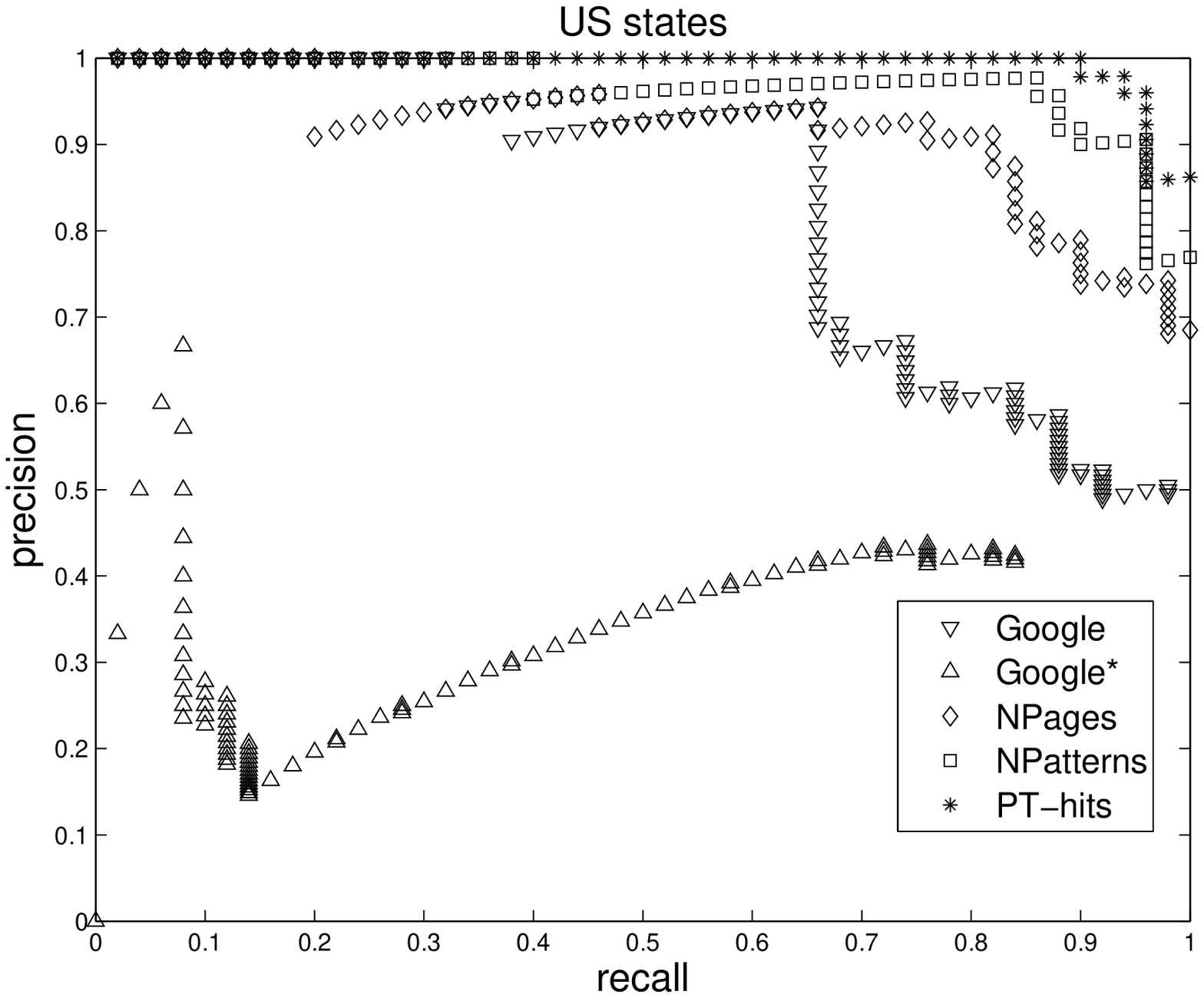}
\hspace{0.01\textwidth} &\includegraphics[width=0.49\textwidth]{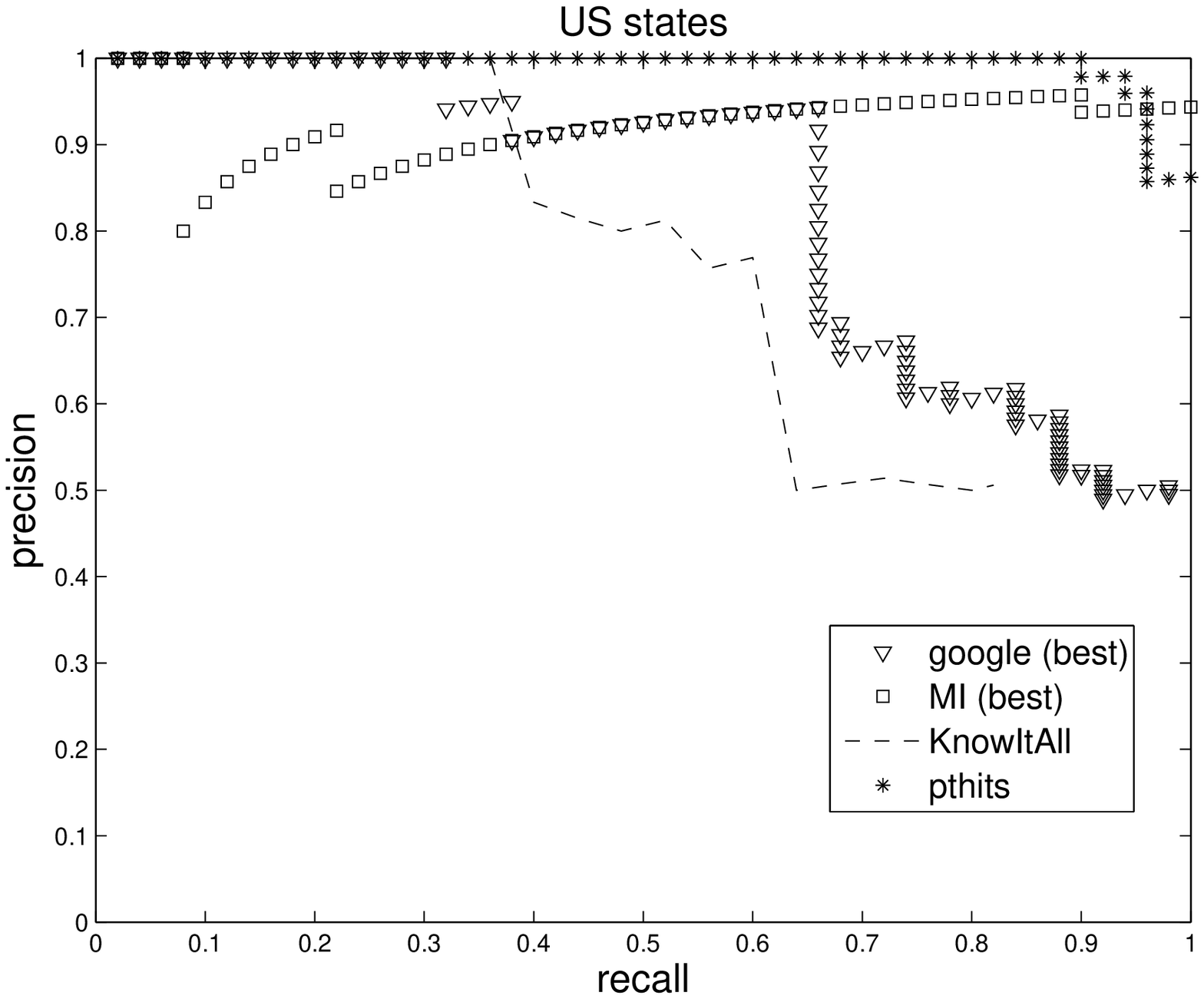}\\ 
(a) \hspace{.3in} & (b)\\ 
\includegraphics[width=0.47\textwidth]{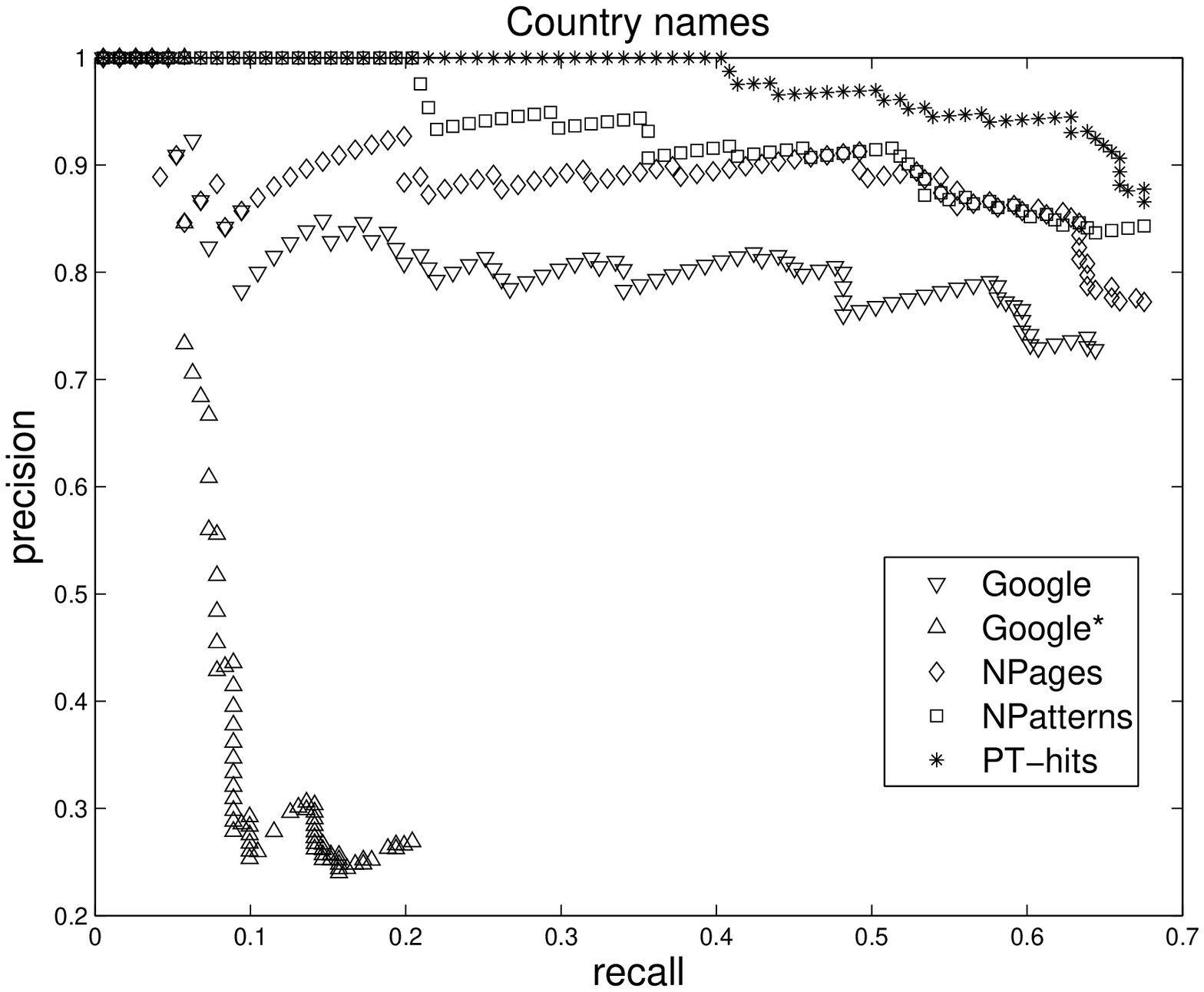} 
\hspace{.01\textwidth} &
\includegraphics[width=0.47\textwidth]{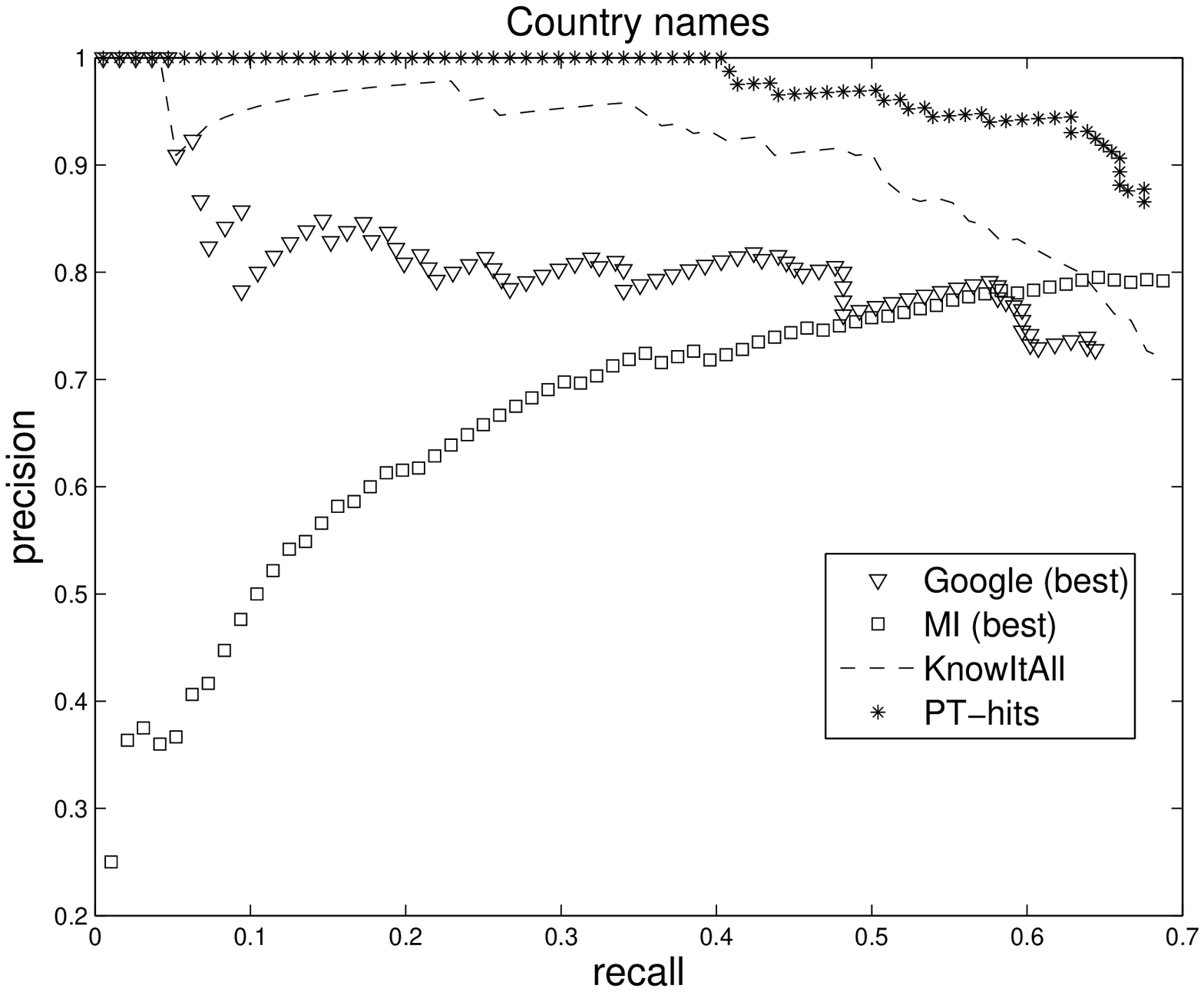}\\
(c) \hspace{.3in} & (d)\\
\end{array} $
\caption{Precision and recall with the extraction target set to the US states
in (a) and (b) and country names in (c) and (d)}
\label{fig:precision}
\end{figure*}

Figures~\ref{fig:precision} (c) and (d) show the comparison between the
same systems with the extraction target set to the country names, 
compiled by the US State
Department\footnote{\url{www.state.gov/www/regions/independent_states.html}},
as the ground truth. The query ``countries such as \%'' was used for the
extraction in DeWild. Many of the systems except MI perform very well, 
with PT-hits again showing the best performance.
This could be partly because of the
high quality of search engine results for this particular query.
MI performs the least favorable for low recalls, and this could be
due to the inaccuracy of the search engine results to our count queries.
 
\subsection{Number of Rewritings}
Adding each query rewriting introduces some cost at the query processing 
time, and a question is if this additional cost is justified.
To evaluate the effect of the number of rewritings on the precision,
we used PT-hits to compile a list of ``US states'' but varied the number
of rewritings that were used. We chose random sets of 2, 3, and 4 patterns
from Table~\ref{tab:countries-states},
and ran PT-hits each time with only one of these sets.
Each experiment was repeated three times, each with a different random
pattern selection. Figure~\ref{fig:varyingNumOfPatterns} shows the
average precision at each position.
The average precision improves when the number of patterns increases 
from 2 to 3. The precision further improves significantly when 
the number of patterns is increased from 3 to 4. Recall can either 
stay the same or increase using more rewritings but can never deteriorate.
\begin{figure}[htbp]
\centering
\includegraphics[width=0.46\textwidth]{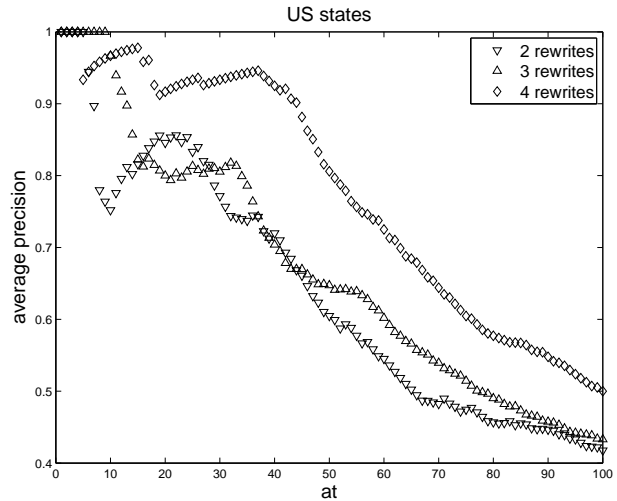}
\caption{Average precision at each position for PT-hits, varying the number 
of patterns for target US states}
\label{fig:varyingNumOfPatterns}
\end{figure}
\subsection{TREC Questions}
Finding candidate answers for natural language questions in a closed corpus can be
challenging, especially if answers are phrased differently than questions, and this
is usually the case in the question answering (QA) track of TREC.
A solution adopted by some QA systems (e.g. AskMSR~\cite{bdb02}) is to find question answers
in a larger collection such as the Web and project the answers to a TREC 
corpus to verify their correctness. 
If a question is formulated as a DeWild query, we can use 
our approach to locate the answer from the Web. 

For our evaluation, we took the first five QA targets from the TREC 2004
dataset~\cite{voorhees04overview}; since a QA target consisted of multiple 
questions, we ended up with a total of 22 questions in the experiment.
For each question, we report the number of correct answers given by TREC,
the rank of the first TREC answer under DeWild, and the number of overlaps 
between TREC and DeWild.
The result of the evaluation is presented in Table~\ref{tab:trec22q}. 

For 64\% of the questions (i.e. 14 questions), all answers returned by TREC 
were also returned by DeWild in positions 1 to 3.
For 18\% of the questions, either we didn't have an obvious formulation of the questions
or our formulation did not retrieve any matching candidate to TREC.
These questions are marked with ``na'' in the table. 
We found out that the TREC answers for question 1.3 were not the ground truth
on the Web; therefore there was a small overlap between TREC and DeWild.
For question 5.4, which asked for the CEO of AARP, TREC had ``Horace Deets'' 
or ``Tess Canja'' as the correct answer; this was based on the information 
in year 2004. At the time
of running our experiments, the correct answer was ``Bill Novelli'', 
and DeWild extracted the more up-to-date correct answer.

\begin{table}[htbp]
\centering
\begin{tabular}{|r|r|r|r|} \hline
Question & TREC & \multicolumn{2}{c|}{DeWild results}   \\
         & ans. & rank(1st overlap) & \# of overlaps \\ 
\hline \hline
1.1 & 1 & 2 & 1 \\ 
1.2 & 1 & 2 & 1 \\ 
1.3 & 14& 3 & 5 \\
1.4 & 1 & na & na \\
1.5 & 1 & 1 & 1 \\
2.1 & 1 & 1 & 1 \\
2.2 & 1 & 3 & 1 \\
2.3 & 5 & 1 & 3 \\
2.4 & 1 & 3 & 1 \\ 
3.1 & 1 & 2 & 1  \\
3.2 & 1 & na & na \\
3.3 & 1 & na & na \\
4.1 & 1 & 1 & 1 \\
4.2 & 1 & 1 & 1 \\
4.3 & 1 & 1 & 1 \\
4.4 & 4 & 2 & 3 \\ 
4.5 & 1 & 1 & 1 \\
5.1 & 1 & 1 & 1 \\ 
5.2 & 1 & 1 & 1 \\ 
5.3 & 1 & 2 & 1 \\
5.4 & 1 & na & 0 \\ 
5.5 & 6 & 9 & 2 \\ 
\hline\end{tabular}
\caption{DeWild's handling of the first 22 questions from TREC 2004}
\label{tab:trec22q}
\end{table}

DeWild sometimes returned additional instances of which some were correct
and others were incorrect but appeared with the query and gave additional 
information. For instance, consider the question ``Who discovered prions?''
from TREC which has only one correct answer.  We transformed the question 
to ``prions are discovered by \%'' and passed it as a query to DeWild.
The highest ranked instance, ``Stanley Prusiner'', was the correct answer to 
the question, and it also received a substantially larger weight than the 
second best instance. Our system returns other acceptable 
answers, including the 8th-ranked ``Dr. Stanley Prusiner'', the 9th-ranked 
``researcher Stanley Prusiner'', and the 12-th ranked ``Nobel Prize winner 
Stanley Prusiner''.  These other answers show that Stanley Prusiner was a 
doctor, a researcher, a Nobel prize winner and he was from the University of 
San Francisco. 

\subsection{Ad Hoc Data Extractions}
As our last experiment, we tried to compile useful resource lists which
we could not find in a list format anywhere on the Web. 
In one case, we tried to find the names of summer 
movies. Although some online resources maintain a quite complete list of 
movies, they don't classify movies as summer movies or otherwise. 
The pattern ``\% is a summer *blockbuster*'' is used as the query for the task. 
The term {\it blockbuster}, which is enclosed by * wild cards in the query, is
augmented by two extra related terms: {\it movie} and {\it film}. 
We manually evaluated the extracted results using the Internet Movie 
Database (IMDB) and concluded that all the results in the top 10
were indeed correct movie names, and their 
release dates were in the summer.

In one more experiment, we used the query ``\% is a Canadian writer'' to compile
a list of Canadian writers. This time, we put together a set of rewritings that
were specific to the query. The query returned over 1300 names. 
We could verify that 91 of the first 100 rows were real Canadian writers.
Of the first 200 rows that we verified, 156 were real Canadian writers.
We also compared the 
first 200 tuples to two of the most comprehensive online lists of Canadian writers
that we could find. DeWild retrieved 86 real author names which could not be 
found in one list\footnote{\url{www.track0.com/ogwc/authors}}
and 70 names which were not in the other 
list\footnote{\url{www.umanitoba.ca/canlit/authorlist}}. 
After combining the two lists, 
DeWild still reported 58 names which we couldn't find in the combined list.
This experiment shows that our queries can be used to compile a reasonably 
good list of resources which can be further edited for correctness.

\section{Related Work}
\label{sec:related}

There is a large body of work on question answering. 
Many systems use a combination of NLP techniques (deep or shallow), learning
algorithms and hand-crafted rules to classify the questions and to establish
relationships between terms of a question and a possible answer sentence
(e.g. \cite{kwok01scaling,askmsr02}). 
Despite some overlap, there are fundamental differences between our work and
the work on question answering. Our work is not a replacement for QA systems
that can often handle complex questions; it is rather a natural way of
expressing short questions or extraction targets.
It is possible to integrate our work 
within a question answering system if natural language questions can be
mapped to DeWild queries.

Large-scale data extraction from the Web has been the subject of various 
recent work~\cite{Eugene-thesis}. In particular, Brin~\cite{Brin98} 
suggests an algorithm which takes a small number of examples of a class
as a seed set and extracts more examples of the same class.
His algorithm learns a set of extraction patterns for each page (or pages with the same
URL prefix) that contains some of the examples and use those patterns to extract more 
tuples from the same page(s). This algorithm does a good job when data is structured
in a tabular format but is not expected to work on free text. 
This is because it is generally unlikely to find more than one example 
(of the seed set) in a 
text document such that their surrounding texts are the same.

KnowItAll~\cite{knowitallshort} takes the description
of a concept or class (e.g. cities) as input and extracts instances 
(e.g. Paris, New York, \ldots) of the class. The system maintains a set of 
rules which can be instantiated with an input class to produce  
keywords that must be used to extract the instances of the class.
KnowItAll uses co-occurrence statistics, specifically mutual information,
to assess the relatedness of each instance.
Our approach differs from KnowItAll in several important aspects: 
First, the query-based interface and the support of wild cards make 
DeWild more adaptive to different extraction tasks. Second,
unlike KnowItAll 
where a concise description of a class must be given, a DeWild query may
specify only the context in which the instances may appear, and this is useful
when a concise class description is not available (this problem is somewhat
addressed in TextRunner~\cite{textrunner08}).
Last but not least, our approach of porting link-based ranking algorithms
for assessing extraction results from text is novel and 
performs better than the one used in KnowItAll.

Related to our query rewritings is the work on query expansion 
\cite{Billoti04}
and query transformation for question answering \cite{ALG01}.
Query expansion has shown to be difficult for phrase queries.
Also query expansion and transformation techniques are
not directly applicable to queries with wild cards.
Our queries can benefit from inverted indexes on terms, phrases,
and N-grams and there are already some works in these areas (e.g. \cite{CaEt05,FREE}).
Other related work includes the work on Web query languages and wrappers
(see \cite{FLM98} for a survey of this area before 1998).
These works can be used to extract data from a specific site or 
a set of pages with similar structures but are not generally applicable to free text. 
Finally, Google's {\it fill-in-the-blank} is related to our wild cards but is different.
Google returns a ranked list of pages for a fill-in-the-blank search but
the ranking is different (and the detail is not published).
Our evaluation includes Google's fill-in-the-blank (referred to as 
Google*).

\section{Conclusions}
\label{sec:con}

We presented a framework for querying and large-scale data extraction from 
natural language text, and evaluated the effectiveness of our framework
within a few data extraction tasks on the Web. We analyzed our rankings of
the results in terms of the stability and the locality of the scoring functions
and conducted experiments comparing their effectiveness in terms of precision 
and recall. Our querying
interface is intuitive for writing queries and scalable to large 
number of rewritings and with more wild cards.

Our work opens a few interesting directions for future work.
First, it would be interesting to study other wild cards, querying
schemes and formalisms that are simple for writing queries and 
at the same time have a well-defined syntax and semantics for 
query evaluations and optimizations. We consider our wild card
queries as a first step toward more formal querying of natural language text.
Second area for possible future work is data storage and indexing. 
Despite the extensive work on indexing free text, there is not much
work on indexing natural language text in particular. 
Given that natural language text can be parsed, there is much room 
for research on building indexes that are aware of sentence structures
and queries. 
Third, a more formal query syntax and semantics opens the door for more study 
on mapping queries to evaluation plans and access path selection and 
optimization.
These are important issues when querying large repositories and/or posing 
complex queries.
Another area for further study is on
extracting {\it n-ary} relations for $n>3$; the problem in general 
is difficult since the columns of target rows can be scattered in 
multiple sentences. 
Yet one more area is on mapping natural language
questions to more formal queries that can be efficiently evaluated.

\section*{Acknowledgments}
\noindent This work is supported by Natural Sciences and Engineering
Research Council of Canada.

\bibliography{dewild} 
\bibliographystyle{abbrv}
\appendix
We give two definitions before proving our theorems.
\begin{definition}
{\rm The normalized {\it Kendall tau} distance between two $n$-dimensional
real vectors $w_1$ and $w_2$ is defined as}
\[ 
d_{kt}(w_1,w_2) = \frac{2}{n(n-1)} \sum_{i=1}^{n} \sum_{j=1}^{n} I_{w_1,w_2}(i, j)
\]
{\rm where} 
\[ I_{w_1,w_2}(i, j) = 
\left\{\begin{array}{ll}
1 & if\ w_1(i) < w_1(j)\ AND\ w_2(i) > w_2(j)\\
0 & otherwise.  
\end{array} \right.
\]
\end{definition}

The Kendall tau distance measures the number of pairwise disagreements between
two (ranked) lists. The Kendall tau distance changes if there is a change in
ordering, but any change in the actual scores would not affect the Kendall tau
distance if the ordering remains the same.

\begin{definition}
{\rm The normalized {\it Manhattan} distance between two $n$-dimensional
real vectors $w_1$ and $w_2$ is defined as}
\[
d_1 (w_1,w_2) = \frac{1}{n}\sum_{i=1}^{n} |w_1(i) - w_2(i)|.
\]
\end{definition} 
{\em Theorem \ref{thm:npatt-local}.
The scoring function NPatterns is stable under both $d_{kt}$ and $d_1$.}
\begin{proof}
Let $G^\prime = G\backslash \{e_1,\ldots, e_k\}$ and $\{1,\ldots,l\}$
be the set of tuples in $G^\prime$ which are affected by the change. 
Clearly $l \leq k$.

\noindent 
{\it (Case for $d_{kt}$)} 
Consider the ordering of the tuples under NPatterns(G).
With every edge $e_i$ removed from the graph, the ordering can be disturbed
by $n$ substitutions at most. This extreme case arises when all tuples have
the same scores and removing $e_i$ moves a tuple to
the bottom of the list. In total, there are at most $k$ tuples affected.
Therefore the maximum for $d_kt(w_1,w_2) \leq 2 (kn) / n(n-1) = 2k/(n-1)$,
and $\lim_{n\rightarrow \infty} d_{kt}(w_1,w_2) = 0$.

\noindent 
{\it (Case for $d_1$)}
The sum of the changes in scores for affected nodes cannot exceed $k$
whereas this sum for non-affected nodes is 0. 
Thus $\lim_{n \rightarrow \infty} d_1(w_1,w_2) = 
\lim_{n \rightarrow \infty} k/n = 0$. 
\end{proof}
{\em Theorem \ref{thm:npages-local}.
NPages is stable under $d_{kt}$. NPages is also stable under $d_1$ if
$w(p \rightarrow t)$ bound to a fixed constant.}
\begin{proof} 
{\it (Case for $d_{kt}$)}
The stability proof is similar to the one given for NPatterns.
With every edge $e_i$ removed from the graph, the ordering can be disturbed
by $n$ substitutions at most. With at most $kn$ substitutions in total,
$\lim_{n\rightarrow \infty} d_{kt}(w_1,w_2) = 0$.

\noindent {\it (Case for $d_{kt}$)}
Suppose $w(p \rightarrow t)$ is bound to a constant $c$.
The sum of the changes in scores for affected nodes cannot exceed $ck$,
and the rest follows.
\end{proof}

\end{document}